\documentclass{article}[11pt]

\usepackage{graphicx}
	\usepackage{fullpage}
	\usepackage[colorlinks=true]{hyperref}
\usepackage{xspace,amsmath,amssymb,amsthm}
\usepackage{algorithmicx,algorithm}
\usepackage[noend]{algpseudocode}

\algnewcommand\algorithmicinput{\textbf{INPUT:}}
\algnewcommand\INPUT{\item[\algorithmicinput]}
\algnewcommand\algorithmicoutput{\textbf{OUTPUT:}}
\algnewcommand\OUTPUT{\item[\algorithmicoutput]}

\newcommand{\problem}[1]{\textsc{#1}\xspace}
\newcommand{\cclass}[1]{\ensuremath{\mbox{\textup{#1}}}\xspace}
\newcommand{\FVS}{\problem{Feedback Vertex Set}}
\newcommand{\ASAT}{\problem{Almost 2-SAT}}
\newcommand{\VC}{\problem{Vertex Cover}}

\newcommand{\prev}{\mathrm{prev}}

	\newtheorem{theorem}{Theorem}
	\newtheorem{lemma}{Lemma}
	
	\theoremstyle{definition}
	\newtheorem{definition}{Definition}

\title{Linear-time Kernelization for Feedback Vertex Set}

	\author{Yoichi Iwata\\ National Institute of Informatics\\ \texttt{yiwata@nii.ac.jp}}

\date{}

\begin{document}
\maketitle

\begin{abstract}
In this paper, we propose an algorithm that, given an undirected graph $G$ of $m$ edges and an integer $k$,
computes a graph $G'$ and an integer $k'$ in $O(k^4 m)$ time such that (1) the size of the graph $G'$ is $O(k^2)$,
(2) $k'\leq k$, and (3) $G$ has a feedback vertex set of size at most $k$ if and only if $G'$ has a feedback vertex set of size at most $k'$.
This is the first linear-time polynomial-size kernel for \FVS.
The size of our kernel is $2k^2+k$ vertices and $4k^2$ edges, which is smaller than the previous best
of $4k^2$ vertices and $8k^2$ edges.
Thus, we improve the size and the running time simultaneously.
We note that under the assumption of $\cclass{NP}\not\subseteq\cclass{coNP}/\cclass{poly}$, \FVS does not admit an $O(k^{2-\epsilon})$-size kernel for any $\epsilon>0$.

Our kernel exploits \emph{$k$-submodular relaxation}, which is a recently developed technique for obtaining efficient FPT algorithms for various problems.
The dual of $k$-submodular relaxation of \FVS can be seen as a half-integral variant of $A$-path packing,
and to obtain the linear-time complexity, we propose an efficient augmenting-path algorithm for this problem.
We believe that this combinatorial algorithm is of independent interest.

A solver based on the proposed method won first place in the 1st Parameterized Algorithms and Computational Experiments (PACE) challenge.

\end{abstract}

\section{Introduction}

\subsection{FPT Algorithms and Kernels}
In the theory of parameterized complexity, we introduce parameters to problems and analyze the complexity with respect to both the input length $n=|x|$ and the parameter value $k$.
If an algorithm runs in $f(k)n^{O(1)}$ time for any input of length $n$ and a parameter $k$, it is called a \emph{fixed-parameter tractable (FPT) algorithm}.
If the $n^{O(1)}$ factor is linear, it is called a \emph{linear-time FPT}.
The typical goal of parameterized algorithms is to develop FPT algorithms with a small $f(k)$ (e.g., $c^k$ for a small constant~$c$) and a small $n^{O(1)}$ (e.g., linear in $n$).
Although there are many algorithms that have been developed with smaller $f(k)$ \emph{or} $n^{O(1)}$, achieving the smallest $f(k)$ and $n^{O(1)}$ \emph{simultaneously} is a very difficult task,
and the smallest $f(k)$ factors and the smallest $n^{O(1)}$ factors are often achieved by different algorithms.
Moreover, when trying to improve the $f(k)$ factor, the $n^{O(1)}$ factor is often ignored by using the $O^*$ notation, which hides factors polynomial in $n$,
and when trying to improve the $n^{O(1)}$ factor, the $f(k)$ factor is often ignored by assuming $k$ is a constant.

If we take into account not only whether $f(k)$ is a single exponential ($f(k)=c^k$) or not, but also the base of
exponent $c$, achieving the smallest $c$ and the smallest $n^{O(1)}$ simultaneously becomes much more difficult.
For example, in recent papers, Iwata, Oka, and Yoshida~\cite{DBLP:conf/soda/IwataOY14}, and Ramanujan and Saurabh~\cite{DBLP:conf/soda/RamanujanS14}
have independently obtained $O(4^k m)$-time algorithms for \ASAT, which is a parameterized version of \problem{Max 2-SAT} where a parameter is the number of unsatisfied clauses;
on the other hand, when allowing the $n^{O(1)}$ factor to be super-linear, there exists an $O^*(2.32^k)$-time algorithm~\cite{DBLP:journals/talg/LokshtanovNRRS14}.
These two algorithms are not comparable: the former runs faster when the input is large but the latter runs faster when the parameter is large.
Typically, only algorithms with the smallest $f(k)$ factor or the smallest $n^{O(1)}$ factor have been studied.
However, if there were three algorithms running in time $O(8^k n)$, $O(4^k n^2)$, and $O(2^k n^3)$, all of them are
incomparable: the first is fastest when $4^k<n$, the second is fastest when $2^k<n<4^k$, and the third is fastest when $n<2^k$.
Do we need to develop algorithms with the smallest possible $f(k)$ factor for each $n^d$?
We observe that \emph{kernelization}, which is another basic research object of the parameterized complexity, is useful for avoiding this Pareto optimality.

A kernelization algorithm (or \emph{kernel}) for a parameterized problem is an algorithm that, given an instance $(x,k)$ in time polynomial in $n=|x|$ and $k$,
returns an equivalent instance $(x',k')$ of the same problem
such that $k'\leq k$ and $|x'|\leq g(k)$ for some function $g$.
When the $n^{O(1)}$ factor in the running time is linear in $n$ (i.e., $k^{O(1)}n$), it is called a \emph{linear-time kernel}.
If there is a kernel, by solving the reduced instance exhaustively, we can obtain an FPT algorithm.
Actually, the converse is also true; if there exists an $f(k)n^{O(1)}$-time FPT algorithm, there also exists a kernel of size $f(k)$.
On the other hand, the existence of a polynomial-size (i.e., $|x'|\leq k^{O(1)}$) kernel is non-trivial and, actually,
there are known to exist FPT problems which (unconditionally) do not have any sub-exponential-size kernels~\cite{DBLP:journals/jcss/BodlaenderDFH09}.
As in the case of FPT algorithms, the typical goal is to develop kernelization algorithms with a small size $g(k)$ (e.g.,
linear in $k$) and a fast running time (e.g., linear in $n$).

Although several simple kernels, including the $2k$-vertex kernel for \VC~\cite{NemhauserT75}, are linear-time kernels,
compared with linear-time FPT algorithms, there are only a small number of studies for linear-time polynomial-size kernels.
Examples include \problem{$d$-Hitting Set}~\cite{DBLP:journals/algorithmica/Bevern14},
\problem{Dominating Set} on planar graphs~\cite{DBLP:conf/iwpec/BevernHKNW11,DBLP:conf/iwpec/Hagerup11},
and \problem{$n - k$ Clique Covering}~\cite{DBLP:conf/wg/ChorFJ04}.
There are two reasons for this.
First, obtaining a polynomial-size kernel is already more difficult than obtaining FPT algorithms; there are many problems in FPT for which no polynomial-size kernels are known.
Second, when assuming the parameter $k$ is a constant, which is often done when studying linear-time FPT algorithms,
kernels become uninteresting because we cannot distinguish between $f(k)$ and $k^{O(1)}$.

Nevertheless, improving the $n^{O(1)}$ factor in the running time of kernels is very important because such kernels
can be used as preprocessing for FPT algorithms.
Let us assume that we have a $k^{O(1)}n^d$-time polynomial-size kernel and an $f(k)n^{O(1)}$-time FPT algorithm.
Then, by applying the FPT algorithm against the instance reduced by the kernel, we obtain a
$k^{O(1)}(f(k)+n^d)$-time FPT algorithm.
Thus, the $n^{O(1)}$ factor of any FPT algorithm can be replaced by $n^d$.
Therefore, if we have a linear-time polynomial-size kernel, we obtain a linear-time FPT algorithm that
simultaneously achieves the smallest possible $f(k)$ factor (ignoring factors polynomial in $k$).
This also implies that after obtaining a linear-time polynomial-size kernel, we can safely ignore the $n^{O(1)}$ factor
and focus on improving the $f(k)$ factor only.
Moreover, it can also be combined with another kernel of smaller size.
Let us assume that we have a $k^{O(1)}n^d$-time polynomial-size kernel and a $g(k)$-size kernel.
Then, by applying the second kernel against the instance reduced by the first kernel, we obtain a $k^{O(1)}n^d$-time
$g(k)$-size kernel.
Therefore, in contrast to the case of FPT algorithms, we can always achieve the smallest size and the fastest running time
simultaneously.

\subsection{Our Contribution}

In this paper, we propose a linear-time quadratic-size kernel for \FVS.
This is the first linear-time polynomial-size kernel for this problem.
\FVS is a problem to decide whether a given undirected graph has a vertex set of size at most a given parameter $k$
whose removal makes the graph a forest.
\FVS is one of the most comprehensively studied problems in the field of parameterized complexity and many different FPT algorithms and
kernels have been developed.
Moreover, the problem was chosen as a target problem of the 1st Parameterized Algorithms and Computational Experiments (PACE) challenge\footnote{\url{https://pacechallenge.wordpress.com/}}.
Actually, this research is strongly motivated by the PACE challenge.
The proposed methods are easy to implement, and a solver\footnote{\url{https://github.com/wata-orz/fvs}} based on the proposed methods won first place in the challenge.

The first FPT algorithm for \FVS was given by Downey and Fellows~\cite{DBLP:conf/dagstuhl/DowneyF92}.
This algorithm and subsequent improved algorithms~\cite{DBLP:conf/iwpec/KanjPS04,DBLP:journals/talg/RamanSS06} use the
strategy to branch on short cycles and the $f(k)$ factor of the running time is not a single exponential in $k$.
The first single-exponential FPT algorithms were obtained independently by
Dehne~et~al.~\cite{DBLP:journals/mst/DehneFLRS07} and Guo~et~al.~\cite{DBLP:journals/jcss/GuoGHNW06}, and
several improved algorithms have been
obtained~\cite{DBLP:journals/jcss/ChenFLLV08,DBLP:journals/algorithmica/CaoC015,DBLP:journals/ipl/KociumakaP14}.
The current smallest $f(k)$ factor for deterministic algorithms is $3.62^k$ given by Kociumaka and
Pilipczuk~\cite{DBLP:journals/ipl/KociumakaP14}.
These single-exponential FPT algorithms use the \emph{iterative compression} technique.
For a graph with $n$ vertices and $m$ edges\footnote{If a graph has a feedback vertex set of size at most $k$, we
have $m=O(kn)$.}, a naive implementation of iterative compression requires
$n$ iterations and each iteration takes $f(k)\Omega(m)$ time.
Therefore, the total running time is $f(k)\Omega(nm)$.
For the case of \FVS, by combining it with 2-approximation
algorithms~\cite{DBLP:journals/siamdm/BafnaBF99,DBLP:journals/ai/BeckerG96}, we can solve the problem using only a single
iteration; however, this increases the running time for one iteration to $f(2k)\Omega(m)$.
Thus, for obtaining a linear-time FPT algorithm, the $f(k)$ factor needs to grow from $3.62^{k}$ to $3.62^{2k}$.
When allowing randomness, a simple $O(4^k km)$-time FPT algorithm using random sampling of edges was given by
Becker~et~al.~\cite{DBLP:journals/jair/BeckerBG00}
The current smallest $f(k)$ factor for randomized FPT algorithms is $3^k$ given by
Cygan~et~al.~\cite{DBLP:conf/focs/CyganNPPRW11}
This algorithm uses dynamic programming on tree-decompositions and takes $3^k k^{O(1)}n^2$ time after obtaining a
tree-decomposition of width at most $k$.
As discussed above, by using our linear-time polynomial-size kernel, we can obtain a $k^{O(1)}(3.62^k+m)$-time
deterministic FPT algorithm and a $k^{O(1)}(3^k+m)$-time randomized FPT algorithm.

The first polynomial-size kernel was given by Burrage~et~al.~\cite{DBLP:conf/iwpec/BurrageEFLMR06}
The size of this kernel is $O(k^{11})$, which was improved to $O(k^3)$ by Bodlaender and van
Dijk~\cite{DBLP:journals/mst/BodlaenderD10}, and to $O(k^2)$ by Thomass\'e~\cite{DBLP:journals/talg/Thomasse10}.
Finally, Dell and van Melkebeek~\cite{DBLP:journals/jacm/DellM14} showed that there are no kernels of size
$O(k^{2-\epsilon})$ for any constant $\epsilon>0$ unless $\cclass{NP}\subseteq\cclass{coNP}/\cclass{poly}$.
The size of the current smallest kernel by Thomass\'e is $4k^2$ vertices and $8k^2$ edges.
As discussed above, if there is a linear-time polynomial-size kernel, by combining it with the smallest kernel, we can
achieve the linear running time and the smallest kernel size simultaneously.
However, our linear-time quadratic-size kernel does not rely on such a combination.

Before providing a description of our kernel, we first give a brief description of a key idea behind the existing kernels.
All the existing kernels for \FVS exploit \emph{$s$-flowers}.
A set of simple cycles is called an $s$-flower if each cycle contains the vertex $s$ and none of two cycles share a
vertex different from $s$.
If the degree of $s$ is large and the graph is well-connected, there exists a large $s$-flower.
Because the size of an $s$-flower (i.e., the number of cycles) gives a lower bound of the size of the minimum feedback
vertex set that does not contain $s$, if it is larger than the parameter $k$, we can remove $s$.
Otherwise, the degree of $s$ is small, or the graph is not well-connected.
In the former case, we know that the graph is small, and in the latter case, we can apply another reduction rule.

In our kernel, instead of $s$-flowers, we exploit \emph{$k$-submodular relaxation}, which is a recently developed
technique for obtaining efficient FPT algorithms for various problems.
The concept of $k$-submodular relaxation was independently discovered by Wahlstr\"om~\cite{DBLP:conf/soda/Wahlstrom14}
and by Iwata and Yoshida, and two results are combined in the full version~\cite{DBLP:journals/corr/Wahlstrom13}.
The $k$-submodular relaxation is a technique to obtain \emph{half-integral} and \emph{persistent} relaxations and many
existing half-integral LP relaxations (e.g., the LP relaxation of \VC~\cite{NemhauserT75}) can be re-derived by this
technique.
If a problem admits such a relaxation, the branch-and-bound method gives an FPT algorithm.
By applying $k$-submodular relaxation, Wahlstr\"om~\cite{DBLP:conf/soda/Wahlstrom14} obtained an $O^*(4^k)$-time FPT
algorithm for \FVS.
The detail description of the $k$-submodular for \FVS is given in Section~\ref{sec:pre}.
By exploiting $k$-submodular relaxation, we obtain a very simple kernel for \FVS.
The size of our kernel is $2k^2+k$ vertices and $4k^2$ edges, which is smaller than the previous best
of $4k^2$ vertices and $8k^2$ edges~\cite{DBLP:journals/talg/Thomasse10}.

We observe that there is a strong relationship between the $k$-submodular relaxation of \FVS and $s$-flowers;
the problem of computing a maximum $s$-flower is the integral dual of the $k$-submodular relaxation of \FVS.
This resembles the situation for \ASAT.
For \ASAT, Raman~et~al.~\cite{DBLP:conf/esa/RamanRS11} obtained an $O^*(9^k)$-time FPT algorithm by a reduction to
\problem{Vertex Cover above Maximum Matching}, and then both the $f(k)$
factor~\cite{DBLP:journals/talg/LokshtanovNRRS14} and $n^{O(1)}$
factor~\cite{DBLP:conf/soda/IwataOY14,DBLP:conf/soda/RamanujanS14} were improved by a reduction to \problem{Vertex Cover above LP}.
Here, the maximum matching is the integral dual of the LP relaxation of \VC.
Because the fractional minimum of the primal LP is always at least the integral maximum of the dual LP, by using the
half-integral relaxation instead of the integral dual, we can obtain a better lower bound.
Moreover, by using the half-integral relaxations, we can directly exploit the persistency of the relaxations.

For obtaining linear-time kernel, we need to solve the $k$-submodular relaxation of \FVS efficiently.
Iwata, Wahlstr\"om, and Yoshida~\cite{DBLP:journals/corr/Wahlstrom13} showed that various cases of $k$-submodular
relaxations can be solved efficiently by a reduction to the minimum cut problem, and by using this method, they obtained
linear-time FPT algorithms for various problems including \problem{Unique Label Cover}.
 However, this method cannot be applied to the $k$-submodular relaxation of \FVS for two reasons:
the method can be applied only to edge-deletion problems and, moreover, when applied to \FVS, the size of the network
for the minimum cut problem becomes exponential in $m$.

For solving the $k$-submodular relaxation of \FVS efficiently,
we propose a max-flow-like augmenting-path algorithm.
This is the most technical part of the paper.
Our algorithm can compute a minimum solution in $O(km)$ time.
We note that this algorithm can be used not only for the linear-time kernel but also for improving the $n^{O(1)}$ factor
of the $O^*(4^k)$-time FPT branch-and-bound algorithm for \FVS.
Wahlstr\"om~\cite{DBLP:conf/soda/Wahlstrom14} applied $k$-submodular relaxation to two general
versions, \problem{Subset Feedback Vertex Set} and \problem{Group Feedback Vertex Set}, and obtained $O^*(4^k)$-time FPT
algorithms for both problems.
Extending our algorithm for solving the $k$-submodular relaxation of these two general versions would be a possible
approach for obtaining linear-time FPT algorithms for these problems.
We note that for \problem{Subset Feedback Vertex Set}, an $O(25.6^k m)$-time randomized FPT algorithm and a $(2^{O(k\log
k)}m)$-time deterministic FPT algorithm have been obtained using a very different
approach~\cite{DBLP:conf/icalp/LokshtanovRS15}.

\subsection{Relation to $A$-path Packing}
Because $s$-flower is the integral dual of the $k$-submodular relaxation of \FVS, one may think that, by creating a copy of each vertex, we
can compute the maximum half-integral $s$-flower which would be a fractional dual of the $k$-submodular relaxation.
However, this is not true.
In the fractional dual of the $k$-submodular relaxation, we need to avoid U-turns (see Section~\ref{sec:flower}
for a detailed description of the fractional dual of the $k$-submodular relaxation).
For example, $s$--$a$--$b$--$c$--$d$--$a$--$s$ is a valid cycle\footnote{This cycle contains the vertex $a$ twice but can be packed half-integrally. When focusing on integral packings, this cycle also becomes invalid.}
but $s$--$a$--$b$--$c$--$b$--$a$--$s$ is invalid;
however, after creating a copy $v'$ for each vertex $v$, we cannot distinguish between a valid cycle $s$--$a$--$b$--$c$--$d$--$a'$--$s$ and an
invalid cycle $s$--$a$--$b$--$c$--$b'$--$a'$--$s$.

Given a set of vertices $A$, a path is called an \emph{$A$-path} if the two end points are in $A$.
Although the maximum $s$-flower can be computed by a reduction to the $A$-path packing problem~\cite{DBLP:journals/talg/Thomasse10},
the $k$-submodular relaxation of \FVS cannot be solved by the reduction to $A$-path packing due to the above reason.
We observe that by a reduction to a more general problem, $A$-path packing in group labelled graphs, we can avoid U-turns
and thus the $k$-submodular relaxation of \FVS can be solved.
However, the current fastest algorithm for this general problem~\cite{DBLP:journals/siamdm/Yamaguchi16} is not enough
for obtaining linear-time kernel.
Our augmenting-path algorithm is actually solving a special case of the fractional group-labelled $A$-path packing,
and it might be useful for obtaining a fast algorithm for the general case.

\subsection{Organization}
First, in Section~\ref{sec:pre}, we give definitions used in this paper, introduce common techniques of kernelization of \FVS,
and introduce the $k$-submodular relaxation of \FVS.
In Section~\ref{sec:simple}, we give a simple quadratic-size kernel by exploiting the $k$-submodular relaxation.
In Section~\ref{sec:flower}, we give an $O(km)$-time augmenting-path algorithm for solving the $k$-submodular
relaxation.
Finally, in Section~\ref{sec:linear}, we give a linear-time quadratic-size kernel by combining these two results.

\section{Preliminaries}\label{sec:pre}

\subsection{Definitions}

A \emph{multiplicity function} of a multiset $S$ is denoted by $\mathbf{1}_S$;
e.g., when $S=\{a,a,b\}$, $\mathbf{1}_S(a)=2$, $\mathbf{1}_S(b)=1$, and $\mathbf{1}_S(c)=0$.
Let $f:U\rightarrow\mathbb{R}$ be a function.
For a multiset $S$, we write the sum of $f(a)$ over $a\in S$ as $f(S)=\sum_{a\in S}f(a)$;
e.g., when $S=\{a,a,b\}$, $f(S)=2f(a)+f(b)$.
We denote the preimage of $i\in\mathbb{R}$ under $f$ by $f^{-1}(i)=\{a\in U\mid f(a)=i\}$.

Let $G=(V,E)$ be an undirected graph.
We assume that $G$ may contain a self-loop and multiple edges.
We will often denote the number of vertices by $n$ and the number of edges by $m$.
We denote the set of edges incident to a vertex $v$ by $\delta_G(v)$ and define the \emph{degree} of $v$ as $d_G(v)=|\delta_G(v)|$.
Here, we note that multiple edges contribute to the degree by its multiplicity, and we never refer to the degree of a vertex having a self-loop.
We omit the subscript $G$ if it is clear from the context.
An edge $e\in E$ is called a \emph{bridge} if its removal increases the number of connected components.

For a vertex set $S$, we denote the graph obtained by removing $S$ and their incident edges by $G-S$.
When $S$ is a singleton $\{v\}$, we simply write $G-v$.
A vertex set $S\subseteq V$ is called a \emph{feedback vertex set} if $G-S$ is a forest.
We denote the size of the minimum feedback vertex set of $G$ by $\mathrm{fvs}(G)$.

A \emph{walk} is an ordered list $(v_0,e_1,v_1,e_2,\ldots,v_{l-1},e_l,v_l)$ such that $l\geq 1$ and each edge $e_i$ connects vertices $v_{i-1}$ and $v_i$.
Note that it may contain a vertex or an edge multiple times.
For a walk $W=(v_0,e_1,\ldots,v_l)$, we denote the multiset of vertices appearing on $W$ by $V(W)=\{v_0,\ldots,v_l\}$
and the multiset of edges appearing on $W$ by $E(W)=\{e_1,\ldots,e_l\}$.

\subsection{Basic Reductions}
We introduce basic reductions that have been commonly used in kernelization algorithms for
\FVS~\cite{DBLP:conf/iwpec/BurrageEFLMR06,DBLP:journals/mst/BodlaenderD10,DBLP:journals/talg/Thomasse10}.
The correctness of these reductions is almost trivial.

\begin{enumerate}
  \setlength{\parskip}{0.05cm}
  \setlength{\itemsep}{0.05cm}
  \item If there is a vertex $v$ containing a self-loop, remove $v$ and decrease $k$ by one.
  \item If there is a vertex of degree at most one, remove it.
  \item If there is a vertex of degree two, remove it and connect its two neighbors by an edge.
  \item If two vertices are connected by more than two edges, replace these edges with a double edge.
\end{enumerate}

Note that rule 3 can remove a vertex that is only incident to a double edge; in this case, it creates a self-loop on its neighbor.
These basic reductions never increase the degree of any vertex and can be fully applied in $O(m)$ time.
After the reduction, the obtained graph has no self-loops and has minimum degree at least three.

We will use the following lemma to bound the size of the kernel.
Because this is a general version of the lemma in \cite{DBLP:journals/talg/Thomasse10}, we give a modified proof.

\begin{lemma}[Thomass\'e~\cite{DBLP:journals/talg/Thomasse10}]\label{lem:size}
If a graph without self-loops satisfies both of the following for an integer $d$, the size of the minimum feedback vertex set is larger than $k$:
\begin{itemize}
  \setlength{\parskip}{0.05cm}
  \setlength{\itemsep}{0.05cm}
  \item At least one of $n>dk+k$ or $m>2dk$ holds; and
  \item for any $v\in V$, it holds that $3\leq d(v)\leq d$.
\end{itemize}
\end{lemma}
\begin{proof}
Let us assume that a graph $G=(V,E)$ has a feedback vertex set $S$ of size $k$ and each vertex satisfies $3\leq d(v)\leq d$.
Because the subgraph on $V\setminus S$ is a forest, the number of edges inside $V\setminus S$ is at most $n-k-1$.
Because each vertex has degree at least three, there are at least $3(n-k)-2(n-k-1)=n-k+2$ edges between $S$ and $V\setminus S$.
On the other hand, because each vertex has degree at most $d$, there can exist at most $dk$ edges between $S$ and $V\setminus S$.

When $n>dk+k$, we have $n-k+2>dk$, which is a contradiction.
Thus, the size of the minimum feedback vertex set is larger than $k$.

Because the total degree of $V$ is $2m$ and the total degree of $S$ is at most $dk$, the total degree of $V\setminus S$ is at least $2m-dk$.
Therefore, there must exist at least $2m-dk-2(n-k-1)$ edges between $S$ and $V\setminus S$.
When $n\leq dk+k$ and $m>2dk$, we have $2m-dk-2(n-k-1)>dk$, which is a contradiction.
Thus, the size of the minimum feedback vertex set is larger than $k$.
\end{proof}

In \cite{DBLP:journals/talg/Thomasse10}, a kernel of $4k^2$ vertices and $8k^2$ edges is obtained by applying the
lemma against $d=4k-1$.
In the next section, we obtain a kernel of $2k^2+k$ vertices and $4k^2$ edges by applying the lemma against $d=2k$.

\subsection{k-submodular Relaxation of \FVS}
A walk $W=(v_0,e_1,\ldots,v_l)$ is called an \emph{$s$-cycle} if $v_0=v_l=s$, $v_i\neq s$ for all $i\in\{1,\ldots,l-1\}$, $e_i\neq e_{i+1}$ for any $i\in\{1,\ldots,l-1\}$ (i.e., there are no U-turns), and each edge is contained in the walk at most twice.
For example, walks $(s,e_1,u,e_2,v,e_3,s)$ and $(s,e_1,u,e_2,v,e_3,w,e_4,u,e_1,s)$ are $s$-cycles but a walk $(s,e_1,u,e_2,v,e_2,u,e_1,s)$ is not.
Note that in this definition, we distinguish each of multiple edges;
e.g., if there is only a single edge $e$ between $s$ and $v$, a walk $(s,e,v,e,s)$ is not an $s$-cycle; however, if there is a double edge $\{e_1,e_2\}$ between $s$ and $v$, a walk $(s,e_1,v,e_2,s)$ is an $s$-cycle.

For a graph $G=(V,E)$ without self-loops and a vertex $s\in V$, a function $x: V\rightarrow\mathbb{R}_{\geq 0}$ is called an \emph{$s$-cycle cover}
if it satisfies that (1) $x(s)=0$ and (2) for any $s$-cycle $C$, $x(V(C))\geq 1$.
Note that $V(C)$ is the multiset of vertices on $C$, and therefore if $x(v)=\frac{1}{2}$ holds for a vertex $v$ contained twice in $C$, we have $x(V(C))\geq 1$.
The \emph{size} of an $s$-cycle cover $x$ is defined as $x(V)$, and when the size $x(V)$ is minimum among all the possible $s$-cycle covers,
it is called a \emph{minimum $s$-cycle cover}.

By introducing the idea of $k$-submodular relaxation, Wahlstr\"om~\cite{DBLP:conf/soda/Wahlstrom14} obtained the
following lemma.
\begin{lemma}[Wahlstr\"om~\cite{DBLP:conf/soda/Wahlstrom14}]\label{lem:ksub}
For any graph $G=(V,E)$ without self-loops and $s\in V$, the following holds:
\begin{itemize}
  \item The size of any feedback vertex set of $G$ that does not contain $s$ is at least the size of the minimum $s$-cycle cover.
  \item There exists a minimum $s$-cycle cover that takes values $\{0,\frac{1}{2},1\}$ (half-integrality).
  \item If there exists a minimum feedback vertex set that does not contain $s$,
  then for any half-integral minimum $s$-cycle cover $x$, there also exists a minimum feedback vertex set $S$ such that
  $x^{-1}(1)\subseteq S$ and $s\not\in S$ (persistency).
\end{itemize}
\end{lemma}

Although the problem of computing a minimum $s$-cycle cover has exponential number of constraints,
Wahlstr\"om~\cite{DBLP:conf/soda/Wahlstrom14} showed that we can compute a half-integral minimum $s$-cycle cover in polynomial time by using the ellipsoid method.

\section{Simple Quadratic-size Kernel}\label{sec:simple}

In this section, we propose a simple polynomial-time quadratic-size kernel for \FVS.
By exploiting the persistency of the $k$-submodular relaxation, we propose the following reduction rule called \emph{$s$-cycle cover reduction}.

For a graph $G=(V,E)$, a vertex $s\in V$, and a half-integral minimum $s$-cycle cover $x$, create a graph $G'=(V,E')$ as follows.
Let $X=x^{-1}(1)$ and let $B\subseteq\delta(s)$ be the set of bridges of $G-X$ connecting $s$ and tree components of $G-X-s$.
Then, $G'$ is obtained from $G$ by inserting a double edge between $s$ and each of $v\in X$ and removing the edges $B$.

\begin{lemma}\label{lem:reduction}
For a graph $G=(V,E)$, a vertex $s\in V$, and a half-integral minimum $s$-cycle cover $x$,
let $G'=(V,E')$ be a graph obtained by applying the $s$-cycle cover reduction.
Then, $\mathrm{fvs}(G)=\mathrm{fvs}(G')$ holds.
\end{lemma}
\begin{proof}
($\geq$)
Let $S$ be a minimum feedback vertex set of $G$.
Observe that all the inserted edges are between $s$ and $X=x^{-1}(1)$.
If $s\in S$, $S$ is also a feedback vertex set of $G'$.
Otherwise, from the persistency, we can assume that $S$ contains all the vertices of $X$.
Thus, $S$ is also a feedback vertex set of $G'$.

($\leq$)
Let $S$ be a minimum feedback vertex set of $G'$.
If $s\in S$, $S$ is also a feedback vertex set of $G$.
Otherwise, because all the vertices of $X$ are connected to $s$ by double edges in $G'$, $S$ must contain all of $X$.
Because all the deleted edges are bridges in $G-X$, $S$ is also a feedback vertex set of $G$.
\end{proof}

After applying this reduction, the degree of $s$ can be bounded as the following shows.

\begin{lemma}\label{lem:degree_bound}
For a graph $G=(V,E)$, a vertex $s\in V$, and a half-integral minimum $s$-cycle cover $x$,
let $G'=(V,E')$ be a graph obtained by applying the $s$-cycle cover reduction.
Then, $d_{G'}(s)\leq 2x(V)$ holds.
\end{lemma}
\begin{proof}
First, we show that $x$ is also an $s$-cycle cover of $G'$.
Let us assume that there is an $s$-cycle $C$ of $G'$ such that $x(C)<1$.
Because $x(v)=1$ for $v\in X=x^{-1}(1)$, $C$ contains none of $X$.
Because all the inserted edges are incident to $X$, $C$ is also an $s$-cycle of $G$, which is a contradiction.

For $i\in\{1,2\}$, let $N_i$ denote the set of vertices that are connected to $s$ by edges of multiplicity $i$ in $G'$.
For each $v_i\in N_1$, we define a vertex $w_i$ as follows.

If the edge $sv_i$ is a bridge in $G'-X$, let $C_i$ be the connected component of $G'-X-s$ containing $v_i$.
Because the reduction removes all the bridges between $s$ and tree components, $C_i$ is not a tree.
Thus, there exists an $s$-cycle contained in $C_i\cup\{s\}$ and, therefore, there must exist a vertex $w_i\in C_i$ with $x(w_i)=\frac{1}{2}$.

If $sv_i$ is not a bridge in $G'-X$, there exists a path $P_i$ from $v_i$ to $N_1\setminus \{v_i\}$ in $G'-X-s$.
Fix an arbitrary path $P_i$ and let $w_i$ be the first vertex on the path such that $x(w_i)=\frac{1}{2}$.
Because $x$ is an $s$-cycle cover, there always exists such a vertex.

If $w_i=w_j$ holds for some $i\neq j$, there exists an $s$-cycle $C$ such that $x(C)=\frac{1}{2}$, which is a contradiction.
Therefore, all $w_i$ are distinct.
Thus, we have $d_{G'}(s)=|N_1|+2|N_2|\leq |x^{-1}(\frac{1}{2})|+2|x^{-1}(1)|= 2x(V)$.
\end{proof}

\begin{algorithm}[t]
\caption{Simple quadratic-size kernelization for \FVS}
\label{alg:simple}
\begin{algorithmic}[1]
\Procedure{Kernelize}{$G,k$}
	\While{true}
		\State Apply the basic reductions
		\If{$k<0$}
			\Return NO
		\EndIf
		\If{$n\leq 2k^2+k$ and $m\leq 4k^2$}
			\Return $(G,k)$
		\EndIf
		\If{$\forall v\in V, d(v)\leq 2k$}
			\Return NO\label{line:small_degree}
		\EndIf
		\State Pick a vertex $s$ of degree larger than $2k$
		\State Compute a half-integral minimum $s$-cycle cover $x$
		\If{$x(V)>k$}
			$G\gets G-s$; $k\gets k-1$\label{line:large_x}
		\Else{}
			apply the $s$-cycle cover reduction\label{line:cycle_cover_reduction}
		\EndIf
	\EndWhile
\EndProcedure
\end{algorithmic}
\end{algorithm}

Now, we describe our simple quadratic-size kernelization algorithm (see Algorithm~\ref{alg:simple}).
First, we apply the basic reduction.
If $k$ becomes negative, we return a NO instance.
If the graph becomes small enough, we return it.
If all the vertices have degree at most $2k$, we return a NO instance.
Otherwise, pick an arbitrary vertex $s$ of degree larger than $2k$, and compute a half-integral minimum $s$-cycle cover $x$.
If the size of the $s$-cycle cover is larger than $k$, we remove $s$ and decrement $k$.
Otherwise, we apply the $s$-cycle cover reduction.

\begin{lemma}
Algorithm~\ref{alg:simple} runs in $(k+m)^{O(1)}$ time and correctly computes $(G',k')$ satisfying $k'\leq k$ and $\mathrm{fvs}(G)\leq k\Leftrightarrow\mathrm{fvs}(G')\leq k'$.
The size of $G'$ is at most $2k^2+k$ vertices and $4k^2$ edges.
\end{lemma}
\begin{proof}
It obviously holds that $k'\leq k$ and the size of $G'$ is at most $2k^2+k$ vertices and $4k^2$ edges.
From Lemma~\ref{lem:degree_bound}, after applying the $s$-cycle cover reduction,
the degree of $s$ changes from more than $2k$ to at most $2k$.
Therefore, the number of edges strictly decreases for each iteration.
Thus, it stops in at most $m$ iterations.
Because each iteration can be done in time polynomial in $k$ and $m$, the total running time is also polynomial in $k$ and $m$.

Next, we show the correctness.
By applying Lemma~\ref{lem:size} against $d=2k$, when the maximum degree is at most $2k$ and at least one of $n>2k^2+k$ and $m>4k^2$ holds,
$\mathrm{fvs}(G)>k$ holds.
Thus, we can safely return a NO instance (line~\ref{line:small_degree}).
From Lemma~\ref{lem:ksub}, if $x(V)>k$, there is no feedback vertex set of size at most $k$ that does not contain $s$.
Thus, we can safely remove $s$ (line~\ref{line:large_x}).
The correctness of the $s$-cycle cover reduction follows from Lemma~\ref{lem:reduction}.
\end{proof}

\section{Efficient Computation of a Half-integral Minimum s-cycle Cover}\label{sec:flower}

In this section, we prove the following theorem.
\begin{theorem}\label{thm:main2}
Given a graph $G=(V,E)$ without self-loops, a vertex $s\in V$, and an integer $k$, in $O(km)$ time,
we can compute a half-integral minimum $s$-cycle cover or conclude that there are no $s$-cycle covers of size at most $\frac{k}{2}$.
\end{theorem}

First, we give several definitions.
Let $\mathcal{C}_s$ denote the set of all $s$-cycles.
A function $y:\mathcal{C}_s\rightarrow\mathbb{R}$ is called an \emph{$s$-cycle packing} if 
for any vertex $v\in V\setminus\{s\}$, it holds that $\sum_{C\in \mathcal{C}_s} \mathbf{1}_{V(C)}(v) y(C)\leq 1$.
The \emph{size} of an $s$-cycle packing $y$ is defined as $y(\mathcal{C}_s)$, and when the size $y(\mathcal{C}_s)$ is the maximum among all the possible
$s$-cycle packings, it is called a \emph{maximum $s$-cycle packing}.
Because the problem of finding a maximum $s$-cycle packing is the LP dual of the problem of finding a minimum $s$-cycle cover,
the size of the minimum $s$-cycle cover is equal to the size of the maximum $s$-cycle packing.
Thus, if we can find a pair of an $s$-cycle cover $x$ and an $s$-cycle packing $y$ of the same size,
we can confirm that $x$ is a minimum $s$-cycle cover and $y$ is a maximum $s$-cycle packing.

\begin{definition}\label{def:basic}
A function $f:E\rightarrow\{0,\frac{1}{2},1\}$ is called a \emph{basic $s$-cycle packing} if it satisfies the following three conditions.
\begin{enumerate}
  \setlength{\parskip}{0.05cm}
  \setlength{\itemsep}{0.05cm}
  \item For any $e\in\delta(s)$, $f(e)\in\{0,1\}$.
  \item Each vertex $v\in V\setminus\{s\}$ satisfies exactly one of the following four conditions (see Figure~\ref{fig:types}):
  \begin{enumerate}
    \item $f(e)=0$ for all edges $e\in\delta(v)$ (called type-O);
    \item $f(e)=1$ for exactly two edges $e\in\delta(v)$ and $f(e)=\frac{1}{2}$ for none of $e\in\delta(v)$ (called type-I);
    \item $f(e)=1$ for none of $e\in\delta(v)$ and $f(e)=\frac{1}{2}$ for exactly two edges $e\in\delta(v)$ (called type-H);
    \item $f(e)=1$ for exactly one edge $e\in\delta(v)$ and $f(e)=\frac{1}{2}$ for exactly two edges $e\in\delta(v)$ (called type-T).
  \end{enumerate}
  \item For each vertex $v\in V\setminus\{s\}$ of type-H or type-T, the cycle obtained by following edges of value $\frac{1}{2}$ from $v$
  contains an odd number of type-T vertices.
\end{enumerate}
\end{definition}

We call the cycle in the third condition the \emph{half-integral cycle of $v$}.
The size of a basic $s$-cycle packing $f$ is defined as $\frac{1}{2}f(\delta(s))$.
Figure~\ref{fig:basic} illustrates an example of the basic $s$-cycle packing, where solid lines denote edges of value $1$, and dotted lines denote edges of value $\frac{1}{2}$.

\begin{figure}[t]
  \begin{minipage}{0.49\hsize}
	  \centering
	  \includegraphics[scale=0.5]{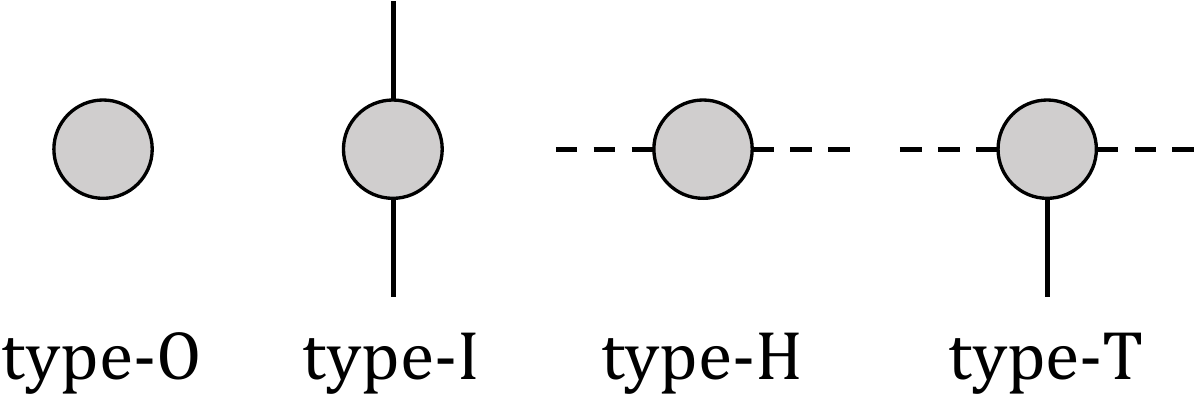}
	  \caption{Four types of vertices}
	  \label{fig:types}
  \end{minipage}
  \begin{minipage}{0.49\hsize}
	  \centering
	  \includegraphics[scale=0.8]{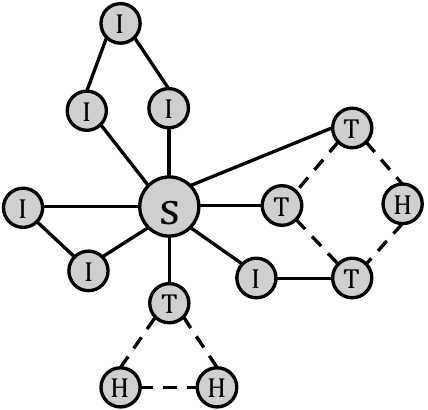}
	  \caption{Example of the basic $s$-cycle packing}
	  \label{fig:basic}
  \end{minipage}
\end{figure}

\begin{lemma}\label{lem:basic_atmost}
If there exists a basic $s$-cycle packing of size $k$, there also exists an $s$-cycle packing of size $k$.
\end{lemma}
\begin{proof}
Given a basic $s$-cycle packing $f$, we construct an $s$-cycle packing $y$ of the same size as follows.
Initially, set $y(C)=0$ for all $s$-cycles $C\in\mathcal{C}_s$.
First, we modify $f$ as follows.
For each type-T vertex $v\in V\setminus\{s\}$, compute a path by following edges of value one from $v$.
If it reaches to $s$, we do nothing.
Otherwise, it reaches to another type-T vertex $u$.
In this case, set $f(e)=0$ for all edges $e$ on the path.
This modification may break the third condition of Definition~\ref{def:basic};
however, it still preserves the first two conditions
($v$ and $u$ change from type-T to type-H and the other vertices on the path change from type-I to type-O)
and does not change the size of $f$.

While $f(\delta(s))>0$, we repeat the following process.
By following edges of value one from $s$, we obtain a (simple) $s$-cycle $C$ or a path from $s$ to a type-T vertex $v$.
In the former case, set $f(e)=0$ for all edges $e$ on $C$ and set $y(C)=1$.
This modification preserves the first two conditions for $f$ (all the vertices on $C$ become type-O).
Because the cycle $C$ is simple, $y$ remains an $s$-cycle packing after the modification.
In the latter case, let $C$ be the half-integral cycle of $v$ and let $\{t_0,\ldots,t_{q-1}\}$ be the type-T vertices on the cycle $C$ in order (the direction is chosen arbitrary from the two).
Observe that each vertex $t_i$ is connected to $s$ by a path $P_i$ consisting of edges of value one.
We use the notation $t_q=t_0$ and $P_q=P_0$.
For each $i\in\{0,\ldots,q-1\}$, let $C_i$ be an $s$-cycle obtained by concatenating the path $P_i$, the path from $t_i$ to $t_{i+1}$ along the cycle $C$, and the path $P_{i+1}$
(when $q=1$, this creates an $s$-cycle obtained by concatenating the path $P_0$, the cycle $C$, and the path $P_0$ again).
For each $s$-cycle $C_i$, set $f(e)=0$ for all edges $e$ on $C_i$ and set $y(C_i)=\frac{1}{2}$.
This modification preserves the first two conditions for $f$ (all the vertices on $C$ or $P_i$'s become type-O).
Because each vertex $v\in V\setminus\{s\}$ is contained in at most two of $C_i$'s, $y$ remains an $s$-cycle packing after the modification.

During the repetition, the sum $\frac{1}{2}f(\delta(s))+y(\mathcal{C}_s)$ does not change.
Thus, when $f(\delta(s))$ becomes zero, the size of $y$ becomes $k$.
\end{proof}

Note that this lemma only says that the size of the maximum basic $s$-cycle packing is always at most the size of the
maximum $s$-cycle packing and does not imply these two are equal;
there might exist an $s$-cycle packing whose size is strictly larger than the size of any basic $s$-cycle packing.
The equality is shown at the end of this section.

\begin{definition}\label{def:augmenting}
For a basic $s$-cycle packing $f$, a walk $W=(v_0,e_1,\ldots,v_l)$ is called an \emph{$f$-augmenting walk} if it satisfies all the following conditions.
\begin{enumerate}
  \setlength{\parskip}{0.05cm}
  \setlength{\itemsep}{0.05cm}
  \item We have $v_0=s$.
  \item We have $f(e_1)=0$.
  \item All the edges $\{e_1,\ldots,e_l\}$ are distinct.
  \item The vertices $\{v_0,\ldots,v_{l-1}\}$ are distinct (the last vertex $v_l$ can be identical to $v_i$ for some $i<l$).
  \item For each $i\in\{1,\ldots,l-1\}$, exactly one of the following holds:
  \begin{enumerate}
    \item $v_i$ is type-O;
    \item $v_i$ is type-I and $f(e)=1$ holds for at least one of $e\in\{e_i, e_{i+1}\}$.\label{cond:typeI}
  \end{enumerate}
  \item If $v_l=v_i$ for some $i<l$, exactly one of the following holds:
  \begin{enumerate}
    \item $v_l=s$ and $f(e_l)=0$;
    \item $v_l$ is type-O;
    \item $v_l$ is type-I and $f(e)=1$ holds for at least one of $e\in\{e_i,e_l\}$.
  \end{enumerate}
  \item If $v_l\not\in\{v_0,\ldots,v_{l-1}\}$, $v_l$ is type-H or type-T.
\end{enumerate}
\end{definition}

For a basic $s$-cycle packing $f$ and an $f$-augmenting walk $W=(v_0,e_1,\ldots,v_l)$, let $f_W:E\rightarrow\{0,\frac{1}{2},1\}$ be a function defined as follows.
First, set $f_W(e)=f(e)$ for all edges $e\in E$.
If $v_l\neq s$ and $v_l=v_i$ holds for some $i<l$, let $h=i$; otherwise, let $h=l$.
Then, for each edge $e\in\{e_1,\ldots,e_h\}$, set $f_W(e)=1-f(e)$.
If $v_l=s$, we finish (see Figure~\ref{fig:augments}-(a)).
Otherwise, we further modify $f_W$ depending on the type of $v_l$.

(Case 1) If $v_l$ is type-O or type-I, for each edge $e\in\{e_{h+1},\ldots,e_l\}$, set $f_W(e)=\frac{1}{2}$ (see Figure~\ref{fig:augments}-(b)).

(Case 2) If $v_l$ is type-H, let $C$ be the half-integral cycle of $v_l$ and let $\{t_0=v_l,t_1,\ldots,t_{q-1}\}$
be the vertex set consisting of the vertex $v_l$ and the type-T vertices on $C$ ordered along $C$
(i.e., $v_l$ is located on the path from $t_{q-1}$ to $t_1$ along the cycle $C$).
We use the notation $t_q=t_0$.
Let $P_i$ be the path from $t_i$ to $t_{i+1}$ along the cycle $C$.
For each even $i$, set $f_W(e)=0$ for all the edges on $P_i$,
and for each odd $i$, set $f_W(e)=1$ for all the edges on $P_i$ (see Figure~\ref{fig:augments}-(c)).

(Case 3) If $v_l$ is type-T, let $C$ be the half-integral cycle of $v_l$ and let $\{t_0=v_l,t_1,\ldots,t_{q-1}\}$
be the type-T vertices on $C$ ordered along $C$.
Then, we proceed in exactly the same way as in case 2 (see Figure~\ref{fig:augments}-(d)).
We note that, in case 2, $q$ is even; thus, $v_l$ is connected to $t_{q-1}$ by edges of value one in $f_W$.
On the other hand, in case 3, $q$ is odd; thus, $v_l$ is connected to none of $t_1$ and $t_{q-1}$.

\begin{figure}[t]
  \centering
  \includegraphics[scale=0.55]{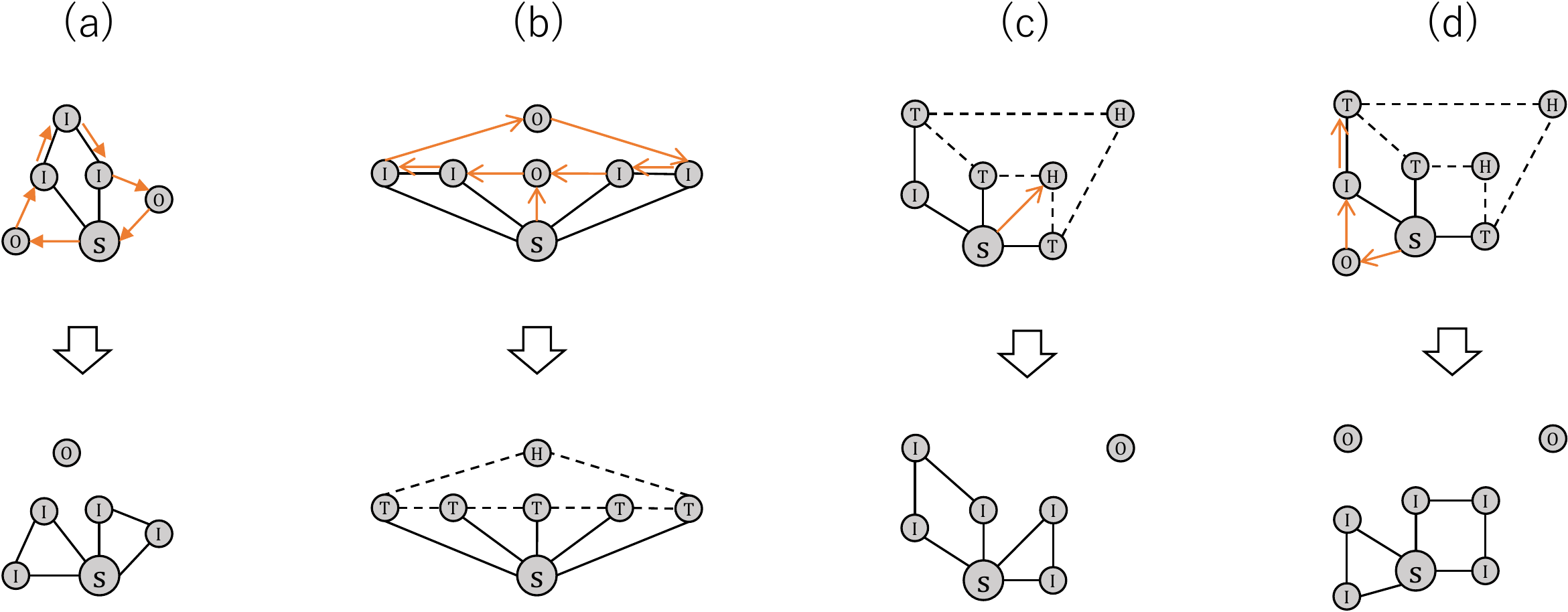}
  \caption{Examples of augmenting walks}
  \label{fig:augments}
\end{figure}

We call this operation that creates $f_W$ from $f$ as \emph{augmenting $f$ along $W$}.

\begin{lemma}\label{lem:augment}
For a basic $s$-cycle packing $f$ and an $f$-augmenting walk $W=(v_0,e_1,\ldots,v_l)$, let $f_W:E\rightarrow\{0,\frac{1}{2},1\}$
be the function obtained by augmenting $f$ along $W$.
Then, $f_W$ is a basic $s$-cycle packing.
Moreover, if $v_l=s$, the size of $f_W$ is the size of $f$ plus one; and otherwise, the size of $f_W$ is the size of $f$ plus $\frac{1}{2}$.
\end{lemma}
\begin{proof}
First, we show the size of $f_W$.
When $v_l=s$, only the edges $e_1$ and $e_l$ are incident to $s$.
Because it holds that $f(e_1)=f(e_l)=0$, $f_W(e_1)=f_W(e_l)=1$, and $f(e)=f_W(e)$ for all the other edges $e\in\delta(s)$, the size of $f_W$ is the size of $f$ plus one.
When $v_l\neq s$, only the edge $e_1$ is incident to $s$.
Because it holds that $f(e_1)=0$, $f_W(e_1)=1$, and $f(e)=f_W(e)$ for all the other edges $e\in\delta(s)$, the size of $f_W$ is the size of $f$ plus~$\frac{1}{2}$.

Next, we prove that $f_W$ is a basic $s$-cycle packing.
For each edge $e\in\delta(s)$, one of $f_W(e)=f(e)$ or $f_W(e)=1-f(e)$ holds.
Thus, the first condition of Definition~\ref{def:basic} is satisfied.

If $v_l\neq s$ and $v_l=v_i$ holds for some $i<l$, let $h=i$; otherwise, let $h=l$.
We show that $f_W$ satisfies the second condition for each vertex $v\in\{v_1,\ldots,v_{h-1}\}$.
From condition 5 of Definition~\ref{def:augmenting}, $v$ must be type-O or type-I in $f$.
If $v$ is type-O in $f$, it changes to type-I in $f_W$.
If $v$ is type-I in $f$, at least one of $f(e_i)=1$ or $f(e_{i+1})=1$ holds from condition 5(b).
If $f(e_i)=f(e_{i+1})=1$, $v$ changes to type-O in $f_W$; otherwise, it remains type-I.

When $v_l=s$, this concludes the proof.
Let us consider the case where $v_l$ is type-O or type-I in $f$.
If $v_l$ is type-O in $f$, it changes to type-T in $f_W$.
If $v_l$ is type-I in $f$, from condition 5(b) and 6(c) of Definition~\ref{def:augmenting}, there are three possibilities:
(1) $f(e_h)=1$ and $\{f(e_{h+1}), f(e_l)\}=\{0,1\}$;
(2) $f(e_h)=1$ and $f(e_{h+1})=f(e_l)=0$; or
(3) $f(e_h)=0$ and $f(e_{h+1})=f(e_l)=1$.
In the first case, $v_l$ changes to type-H in $f_W$, and in the other two cases, it changes to type-T.
From condition 5 of Definition~\ref{def:augmenting}, each vertex $v\in\{v_{h+1},\ldots,v_{l-1}\}$ must be type-O or type-I in $f$.
If $v_i$ is type-O in $f$, it changes to type-H in $f_W$.
If $v_i$ is type-I in $f$, there are two possibilities: $f(e_i)=f(e_{i+1})=1$ or $\{f(e_i),f(e_{i+1})\}=\{0,1\}$.
In the former case, $v_l$ changes to type-H in $f_W$, and in the latter case, it changes to type-T.

Now, let us count the number of type-T vertices in the created half-integral cycle $(v_h,e_{h+1},\ldots,e_l,v_l=v_h)$.
On this cycle, let $\alpha$ be the number of segments of consecutive edges of value one in $f$.
Remember that a vertex $v_i\in\{v_{h+1},\ldots,v_{l-1}\}$ becomes type-T in $f_W$ if and only if it satisfies $\{f(e_i),f(e_{i+1})\}=\{0,1\}$;
in other words, vertices on endpoints of consecutive edges of value one become type-T in $f_W$.
If $f(e_h)=1$ and $\{f(e_{h+1}),f(e_l)\}=\{0,1\}$, $v_l$ becomes type-H in $f_W$ and $2\alpha-1$ vertices in $\{v_{h+1},\ldots,v_{l-1}\}$ become type-T in $f_W$.
Otherwise, $v_l$ becomes type-T in $f_W$ and $2\alpha$ vertices in $\{v_{h+1},\ldots,v_{l-1}\}$ become type-T in $f_W$.
Thus, in both cases, the third condition of Definition~\ref{def:basic} is satisfied.

Finally, let us consider the case where $v_l$ is type-H or type-T in $f$.
If $v_l$ is type-T and $f(e_l)=1$, it changes to type-O in $f_W$.
Otherwise, it changes to type-I in $f_W$.
The other vertices on the half-integral cycle of $v_l$ changes from type-T to type-I, and from type-H to type-O or type-I.
\end{proof}

Now, we give an algorithm to compute an $f$-augmenting walk (see Algorithm~\ref{alg:augment_search}).
First, we initialize a set $S$ and a table $\prev:V\rightarrow E\cup\{\epsilon\}$.
The set $S$ stores vertices we need to process and initialized to $\{s\}$.
We ensure that only the vertex $s$ and vertices of type-O or type-I are stored in $S$.
The table $\prev(v)$ represents an edge to the parent of $v$ in the search tree and initialized to the dummy edge $\epsilon$,
which indicates that the vertex is not visited (or the vertex is the root $s$).
Then, while $S$ is not empty, pick up an arbitrary vertex $u$ from $S$ and process each incident edge $e=uv\in\delta(u)$ as described below.
If $S$ becomes empty, the algorithm returns NO.

First, we check whether the edge $e=uv$ is valid by testing the following three conditions.
If $e=\prev(u)$, because we have already processed this edge, we skip it.
If $e$ is incident to $s$ and $f(e)=1$, because such an edge cannot be used in an augmenting walk, we skip it.
Note that, when $v=s$, at least one of these two conditions are satisfied.
Similarly, if $u$ is type-I and both of $f(\prev(u))$ and $f(e)$ are zero, because we cannot use both of $\prev(u)$ and $e$ simultaneously,
we skip it.

If $v$ is type-H, or type-T, we return the walk from $s$ to $v$ in the search tree by using the table $\prev$.
If $\prev(v)=\epsilon$, we set $\prev(v)=e$ and insert it to $S$.
If $v$ is already visited and $v$ is type-O, let $w$ be the lowest common ancestor of $u$ and $v$ in the search tree.
Then, we return the walk obtained by going down from $s$ to $u$ along the search tree, jumping from $u$ to $v$ by the edge $e$,
and then by going up from $v$ to $w$ along the search tree.
If $v$ is already visited and $v$ is type-I, we basically do the same; however, we need one additional constraint.
If both of $f(\prev(v))$ and $f(e)$ are zero, the walk created as above does not satisfy the condition 6(c) of Definition~\ref{def:augmenting}; thus, we skip the edge without returning the walk.

\begin{algorithm}[t]
\caption{Algorithm for computing an $f$-augmenting walk}
\label{alg:augment_search}
\begin{algorithmic}[1]
\Procedure{FindAugmentingWalk}{$G,s,f$}
	\State $S\gets\{s\}$
	\State $\prev(v)=\epsilon$ for all $v\in V$
	\While{$S\neq\emptyset$}
		\State Pick a vertex $u\in S$ and remove $u$ from $S$
		\For{$e=uv\in\delta(u)$}
			\If{$e=\prev(u)$}
				\textbf{continue}
			\EndIf
			\If{$e\in\delta(s)$ and $f(e)=1$}
				\textbf{continue}
			\EndIf
			\If{$u$ is type-I and $f(\prev(u))=f(e)=0$}
				\textbf{continue}
			\EndIf
			\If{$v$ is type-H or type-T}
				\State $\prev(v)\gets e$
				\State \Return the walk from $s$ to $v$ along the search tree\label{line:augment_search:walk}
			\ElsIf{$\prev(v)=\epsilon$}
				\State $\prev(v)\gets e$; $S\gets S\cup\{v\}$
			\ElsIf{$v$ is type-O or $f(\prev(v))+f(e)\geq 1$}
				\State \Return the walk $s\rightarrow u\rightarrow v\rightarrow w$ along the search tree\label{line:augment_search:walk2}
			\EndIf
		\EndFor
	\EndWhile
	\State \Return NO
\EndProcedure
\end{algorithmic}
\end{algorithm}

From the construction of our algorithm, we obtain the following lemma.
\begin{lemma}
A walk returned by Algorithm~\ref{alg:augment_search} is an $f$-augmenting walk.
\end{lemma}

Note that this lemma does not say that Algorithm~\ref{alg:augment_search} returns an $f$-augmenting walk whenever there exists an $f$-augmenting walk;
it only says that if the algorithm returns a walk, it is an $f$-augmenting walk, and the algorithm might return NO even when there exists an $f$-augmenting walk.
We now show that, if the algorithm returns NO, we can construct an $s$-cycle cover $x$ whose size is equal to the size of $f$.
From Lemma~\ref{lem:basic_atmost} and the LP duality of $s$-cycle packings and $s$-cycle covers, the size of a basic $s$-cycle packing is always at most the size of an $s$-cycle cover.
Therefore, this equality implies that the current basic $s$-cycle packing $f$ is the maximum and the constructed $s$-cycle packing $x$ is the minimum.
This also implies that when the algorithm returns NO, there are no $f$-augmenting walks.

To construct such an $s$-cycle cover $x$, we first prove a property of the table $\prev$.
We call a vertex $v\in V$ \emph{reachable} if $v=s$ or $\prev(v)\neq\epsilon$.
For each edge $e\in\delta(s)$ with $f(e)=1$, by following edges of value 1 from $e$, we can obtain a simple cycle returning to $s$ or a simple path to a type-T vertex.
We denote such a cycle or a path by $W_e$.
Note that when $W_e$ is a cycle, $W_e=W_{e'}$ for another edge $e'\in\delta(s)$.

\begin{lemma}\label{lem:property_prev}
If Algorithm~\ref{alg:augment_search} returns NO, exactly one of the following holds for each edge $e\in\delta(s)$ with $f(e)=1$:
\begin{enumerate}
  \item $\prev(v)=\epsilon$ for any vertex $v\in V(W_e)$;
  \item $W_e$ is a cycle, all the vertices on $W_e$ are reachable, and exactly one vertex $v\in V(W_e)\setminus\{s\}$ satisfies $\prev(v)\not\in E(W_e)$.
\end{enumerate}
\end{lemma}
\begin{proof}
Let us assume that $\prev(v)\neq\epsilon$ holds for at least one of $v\in V(W_e)$.
If a type-I vertex is reachable, the algorithm makes its adjacent type-I vertices connected by edges of value one reachable.
Thus, all the vertices on $W_e$ are reachable.
If $W_e$ is a path, the algorithm makes the type-T endpoint of the path reachable, and therefore the algorithm returns an $f$-augmenting walk at line~\ref{line:augment_search:walk}.
Thus, $W_e$ must be a cycle.
Let $W_e=(v_0=s,e_1,\ldots,v_l=s)$ be the cycle.
If there exists an integer $i\in\{2,\ldots,l-1\}$ such that $\prev(v_{i-1})\neq e_i$ and $\prev(v_i)\neq e_i$, the algorithm returns an $f$-augmenting walk at line~\ref{line:augment_search:walk2}.
Thus, the edges $\{e_2,\ldots,e_{l-1}\}$ must be contained in the set $\{\prev(v_1),\ldots,\prev(v_{l-1})\}$.
Because all the $e_i$'s are distinct, this implies that exactly one of $i\in\{1,\ldots,l-1\}$ satisfies $\prev(v_i)\not\in \{e_2,\ldots,e_{l-1}\}$.
Because $\prev(v)$ cannot be $e_1$ nor $e_l$, we have $\prev(v_i)\not\in E(W_e)$.
\end{proof}

When Algorithm~\ref{alg:augment_search} returns NO, by using the obtained table $\prev$, we construct a function $x: V\rightarrow\{0,\frac{1}{2},1\}$ as follows.
For each edge $e=su\in\delta(s)$ with $f(e)=1$, if $W_e$ is a cycle satisfying the second condition of Lemma~\ref{lem:property_prev},
we set $x(v)=1$ for the unique vertex $v$ satisfying $\prev(v)\not\in E(W_e)$.
Otherwise, we set $x(u)=\frac{1}{2}$.
If $x(u)$ is already set to $\frac{1}{2}$, e.g., $W_{e_1}=W_{e_2}=(s,e_1,u,e_2,s)$ for a double edge $\{e_1,e_2\}$, we set $x(u)=1$.

\begin{lemma}\label{lem:dual}
If Algorithm~\ref{alg:augment_search} returns NO, the function $x$ is a minimum $s$-cycle cover.
\end{lemma}
\begin{proof}
First, we show that $x$ is an $s$-cycle cover.
Let $C=(v_0=s,e_1,\ldots,v_l=s)$ be an $s$-cycle.
If $f(e_1)=f(e_l)=1$ and both of $W_{e_1}$ and $W_{e_l}$ satisfy the first condition of Lemma~\ref{lem:property_prev}, we have $x(C)\geq \frac{1}{2}+\frac{1}{2}=1$.
Otherwise, there are two cases: (1) $f(e)=0$ for at least one of $e\in\{e_1,e_l\}$ or (2) $f(e_1)=f(e_l)=1$ and $W_e$ satisfies the second condition of Lemma~\ref{lem:property_prev} for at least one of $e\in\{e_1,e_l\}$.

(Case 1)
If $f(e_1)=0$ holds (the case of $f(e_l)=0$ is symmetric), the vertex $v_1$ is reachable.
If a vertex $v_i$ is reachable and type-O, the vertex $v_{i+1}$ is also reachable.
If all the vertices $\{v_1,\ldots,v_{l-1}\}$ are type-O, the algorithm returns an $f$-augmenting walk at line~\ref{line:augment_search:walk2},
and if there exists a reachable type-H or type-T vertex, the algorithm returns an $f$-augmenting walk at line~\ref{line:augment_search:walk}.
Therefore, there must exist a reachable type-I vertex on $C$.
Let $v_i$ be the first such vertex and let $W_e$ be the value-one cycle containing $v_i$.
Note that $v_i$ must be contained in a cycle because otherwise it is connected to a type-T vertex by edges of value one and this type-T vertex becomes reachable, which is a contradiction.
If $\prev(v_i)\in E(W_e)$, the algorithm returns an $f$-augmenting walk at line~\ref{line:augment_search:walk2}.
Thus, we have $x(C)\geq x(v_i)=1$.

(Case 2)
If $f(e_1)=f(e_l)=1$ and $W_{e_1}$ satisfies the second condition of Lemma~\ref{lem:property_prev} (the case of $W_{e_l}$ is symmetric), all the vertices on $W_{e_1}$ are reachable.
Let $v_a$ be the first vertex on $C$ satisfying $f(e_{a+1})=0$.
If there is no such vertex, $W_e$ is completely contained in $C$, and therefore, we have $x(C)\geq x(W_{e_1})=1$.
If $\prev(v_a)\not\in E(W_{e_1})$, we have $x(C)\geq x(v_a)=1$.
Otherwise, $v_{a+1}$ is reachable.
Thus, by the same argument as in case 1, there must exist a reachable type-I vertex $v_b$ on $C$ for $b>a$.
Let $W_{e'}$ be the value-one cycle containing $v_b$.
If $\prev(v_b)\in E(W_{e'})$, the algorithm returns an $f$-augmenting walk at line~\ref{line:augment_search:walk2}.
Thus, we have $x(C)\geq x(v_b)=1$.

From Lemma~\ref{lem:property_prev}, the vertex $v$ satisfying $\prev(v)\not\in E(W_e)$ is unique for each cycle $W_e$.
Therefore, the size of $x$ is $\frac{1}{2}f(\delta(s))$, which is equal to the size of $f$.
Thus, $x$ is a minimum $s$-cycle cover.
\end{proof}

\begin{proof}[Proof of Theorem~\ref{thm:main2}]
Because each augmentation increases the size of $f$ by at least $\frac{1}{2}$, after $k+1$ steps, we can obtain a half-integral minimum $s$-cycle cover of size at most $\frac{k}{2}$,
or conclude that there are no $s$-cycle covers of size at most $\frac{k}{2}$.
Because Algorithm~\ref{alg:augment_search} runs in $O(m)$ time, the total running time is $O(km)$.
\end{proof}

Figure~\ref{fig:execute} illustrates an example execution of the augmenting-path algorithm.
Solid orange-colored lines denote edges of value one and dotted orange-colored lines denote edges of value $\frac{1}{2}$.
In each step, the algorithm searches an augmenting walk, which is denoted by arrows in the figure, and augments along the obtained walk.
Finally, when the algorithm fails to find an augmenting walk (see the lower left figure), only vertices $\{s,1,2,4\}$ are reachable from $s$.
The $\prev$ of vertex $1$ is edge $1$--$2$, which is contained in the cycle $s$--$1$--$2$--$s$, and the $\prev$ of vertex $2$ is edge $4$--$2$, which is not contained in the cycle.
Thus, we can construct a function $x$ such that $x(2)=1$, $x(3)=x(5)=x(6)=x(7)=x(8)=\frac{1}{2}$, and $x(1)=x(4)=0$.
This is actually an $s$-cycle cover of the graph, and the size of $x$ is $\frac{7}{2}$, which is equal to the size of the constructed basic $s$-cycle packing.
Therefore, $x$ is the minimum $s$-cycle cover.
\begin{figure}[t]
  \centering
  \includegraphics[scale=0.6]{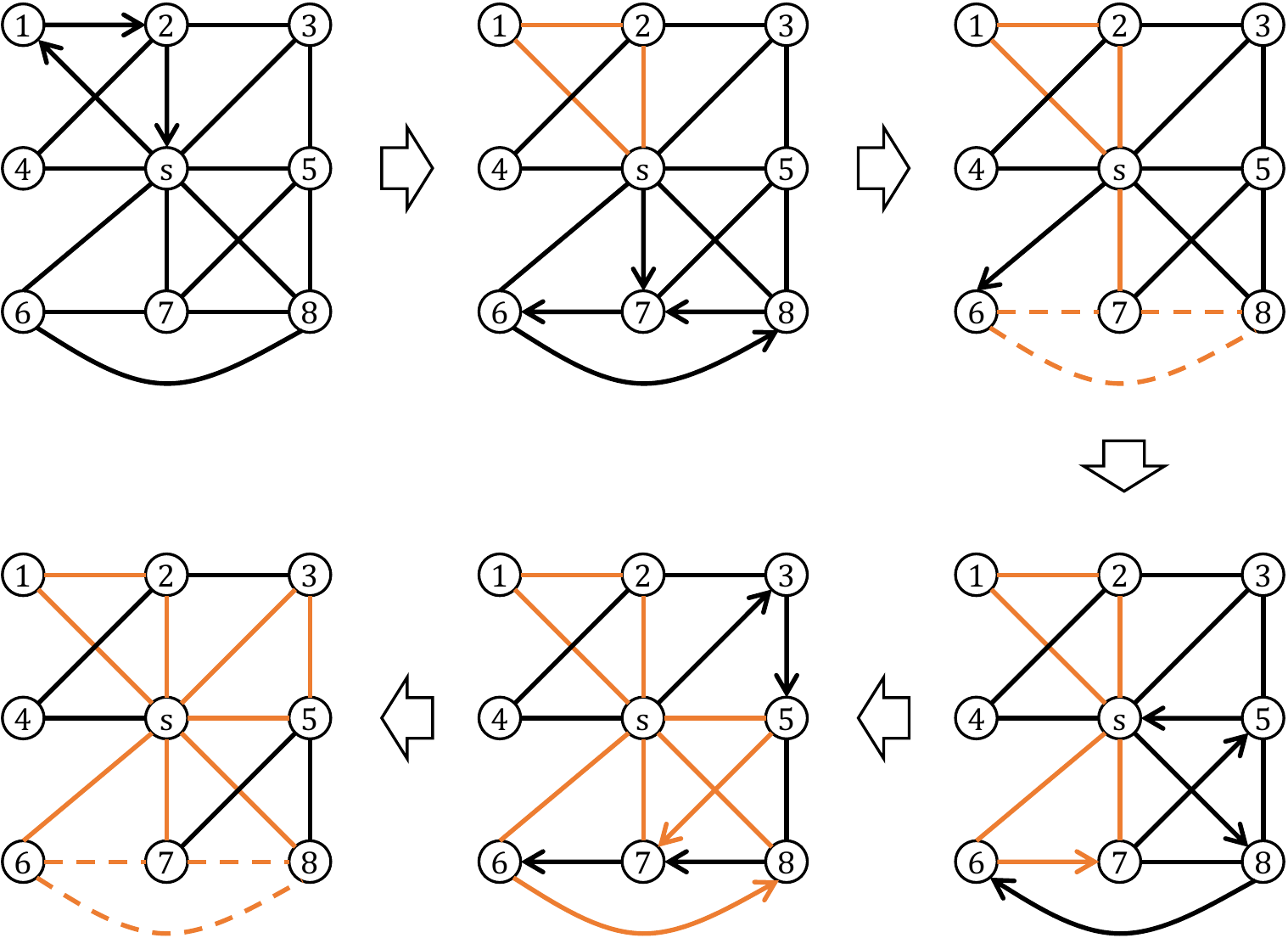}
  \caption{Example execution of the algorithm}
  \label{fig:execute}
\end{figure}

\section{Linear-time Quadratic-size Kernel}\label{sec:linear}

In this section, we improve the running time of the quadratic-size kernel presented in Section~\ref{sec:simple} to $O(k^4 m)$.
By using the $O(km)$-time algorithm for computing the minimum $s$-cycle cover presented in Section~\ref{sec:flower}, each iteration can be done in $O(km)$ time.
However, because the number of iterations is only bounded by $O(m)$, the total running time becomes $O(km^2)$.
We show that, by a slight modification to Algorithm~\ref{alg:simple}, the number of iteration can be bounded by $O(k^3)$;
thus, the total running time becomes $O(k^4 m)$.

We add the following two rules just after line~\ref{line:small_degree} of Algorithm~\ref{alg:simple}.
\begin{itemize}
  \item If there is a vertex $v$ incident to more than $k$ double edges, remove $v$, decrement $k$, and continue the iteration.
  \item If there are more than $k^2$ double edges, return NO.
\end{itemize}

The safeness of these two rules can be shown as follows.
Because any feedback vertex set must contain at least one of the two end points of a double edge,
if there is a vertex $v$ incident to more than $k$ double edges, it must be contained in any feedback vertex set of size at most $k$.
After applying this rule, each vertex can be incident to at most $k$ double edges.
Therefore, any feedback vertex set of size $k$ can delete at most $k^2$ double edges.
Thus, if there are more than $k^2$ double edges, there are no feedback vertex sets of size at most $k$.

For bounding the number of iterations, we use the following lemma.
\begin{lemma}\label{lem:no_double}
For a graph $G=(V,E)$ of minimum degree at least three, a vertex $s\in V$, and a half-integral minimum $s$-cycle cover $x$, if $2x(V)<d_G(s)$ holds, then $x^{-1}(1)\neq\emptyset$.
\end{lemma}
\begin{proof}
From Lemma~\ref{lem:degree_bound}, for a graph $G'$ obtained by applying the $s$-cycle cover reduction, it holds that $d_{G'}(s)\leq 2x(V)$.
When $x^{-1}(1)=\emptyset$, the reduction inserts no new edges and only removes the bridges of $G$ connecting $s$ and tree components of $G-s$.
Because the graph $G-s$ has minimum degree at least two, it has no tree components.
Thus, we have $G'=G$, which is a contradiction.
\end{proof}

Now, we can prove the upper bound on the number of iterations.
\begin{lemma}\label{lem:linear}
The modified Algorithm~\ref{alg:simple} stops in $O(k^3)$ iterations.
\end{lemma}
\begin{proof}
We color each double edge red or blue.
Initially, all the double edges are blue, and after applying the $s$-cycle cover reduction, we color all the double edges incident to $s$ red
(not only newly inserted double edges but also blue colored edges are recolored to red).
The other double edges, which are created by the deletion of degree two vertices in the basic reductions, are colored blue.
Let $\alpha$ denote the number of red double edges and $\beta$ denote the number of vertices of degree larger than $2k$ and incident to at least one red double edge.
Let $k_0$ be the initial value of $k$ and $\phi$ be a potential defined as $\phi=2k_0^2k+4k_0k+3k-2\alpha+\beta$.
Observe that, because red double edges are created only by the $s$-cycle cover reduction, each vertex can be incident to at most $k_0+1$ red double edges,
and that the number of red double edges is always at most $k_0^2+k_0$
(at most $k_0^2$ edges before applying the $s$-cycle cover reduction and the reduction can create at most $k_0$ red double edges).

Initially, there are no red edges; thus, the initial potential is $2k_0^3+4k_0^2+3k_0=O(k_0^3)$.
If $\phi$ becomes negative, we have $k<0$ or $\alpha>k_0^2k\geq k^2$.
Thus, the algorithm returns NO.

When $k$ is decremented, $\alpha$ can decrease by at most $k_0+1$.
Because there are at most $k_0^2+k_0$ red double edges, $\beta$ can increase by at most $2k_0^2+2k_0$.
Thus, $\phi$ decreases by at least
\[
2k_0^2k+4k_0k+3k-2\alpha+\beta-(2k_0^2(k-1)+4k_0(k-1)+3(k-1)-2(\alpha-k_0-1)+(\beta+2k_0^2+2k_0))\geq 1
\]

When applying the $s$-cycle cover reduction, if the reduction creates $c\enspace (\geq 1)$ new red double edges, $\alpha$ increases by $c$ and $\beta$ can increase by at most $c$.
Thus, $\phi$ decreases by at least
\[
2k_0^2k+4k_0k+3k-2\alpha+\beta-(2k_0^2k+4k_0k+3k-2(\alpha+c)+(\beta+c))\geq c\geq 1
\]
After applying the $s$-cycle cover reduction, from Lemma~\ref{lem:no_double}, $s$ is incident to at least one double edge.
Thus, if the reduction does not create any new red double edges, $s$ must be incident to at least one red double edge before the reduction.
From~\ref{lem:degree_bound}, the degree of $s$ becomes at most $2k$ after the reduction.
Therefore $\beta$ decreases by one; thus, $\phi$ decreases by one.

Now, we have shown that each iteration decreases the potential $\phi$ by at least one.
Because $\phi$ is initially $O(k^3)$ and is always non-negative, the number of iterations is $O(k^3)$.
\end{proof}

	\bibliographystyle{abbrv}
\bibliography{paper}

\end{document}